\documentclass[11pt,a4paper]{amsart}
\usepackage{amssymb,amsmath,amsfonts,amsthm,mathtools}
\usepackage{enumerate}
\usepackage[english]{babel}
\usepackage{hyperref}
\usepackage{tikz}
\usepackage{pgfplots}
\newtheorem{theorem}{Theorem}[section]
\newtheorem{lemma}[theorem]{Lemma}

\theoremstyle{definition}
\newtheorem{definition}[theorem]{Definition}
\newtheorem{model}{Stochastic Process}
\theoremstyle{remark}
\newtheorem{remark}[theorem]{Remark}
\newtheorem{example}[theorem]{Example}
\newenvironment{stochprocess}[1]{
	\renewcommand\themodel{#1}
	\model
}{\endmodel}
\newcommand{\RR}{\mathbb{R}}
\newcommand{\EE}{\mathbb{E}}
\newcommand{\CC}{\mathbb{C}}
\newcommand{\NN}{\mathbb{N}}
\newcommand{\ZZ}{\mathbb{Z}}
\newcommand{\cC}{\mathcal{C}}
\newcommand{\cD}{\mathcal{D}}
\newcommand{\cK}{\mathcal{K}}
\newcommand{\eps}{\varepsilon}
\renewcommand{\phi}{\varphi}
\newcommand{\lam}{\lambda}
\newcommand{\indic}{\mathbf{1}}
\newcommand{\hh}{\mathfrak{h}}
\newcommand{\lr}{\left(}
\newcommand{\rr}{\right)}
\renewcommand{\le}{\left[}
\newcommand{\re}{\right]}
\DeclareMathOperator{\Tr}{Tr}

\DeclareMathOperator{\Id}{Id}
\DeclareMathOperator*{\supp}{supp}
\newcommand*\Diff[1]{\mathop{}\!\mathrm{d}#1}
\newcommand{\norm}[1]{\ensuremath{\left\|#1\right\|}}
\newcommand{\bes}{\begin{equation*}}
\newcommand{\ees}{\end{equation*}}
\newcommand{\be}{\begin{equation}}
\newcommand{\ee}{\end{equation}}
\newcommand{\eqs}[1]{\begin{align*}#1\end{align*}}
\newcommand{\eq}[1]{\begin{align}#1\end{align}}

\renewcommand{\div}{\mathrm{div}}
\newcommand{\grad}{\nabla}
\renewcommand{\L}{\Lambda}
\newcommand{\Lip}{\mathrm{Lip}}
\pgfplotsset{compat=1.17}
\pgfplotsset{every axis/.append style={tick label style={font=\footnotesize}}}
\title[Wegner estimate for random divergence-type operators]{Wegner estimate for random divergence-type operators monotone in the randomness}
\subjclass[2010]{Primary 47B80; Secondary 35P15, 35J15, 35R60.}
\keywords{Random divergence-type operators, Wegner estimate, Eigenvalue lifting, Breather type.}
\date{\today}
\hypersetup{
	pdftitle={Wegner estimate for random divergence-type operators monotone in the randomness},
	pdfauthor={Alexander Dicke},
	hidelinks
}
\author[A.~Dicke]{Alexander Dicke}
\address{
	A.~Dicke,
	Fakult\"at f\"ur Mathematik, Technische Univer\-si\-t\"at Dortmund,
	D-44221 Dortmund, Germany
}
\email{alexander.dicke@mathematik.tu-dortmund.de}
\begin{document}
%
%%%%%%%%%%%%%%%%%%%%%%
%%%%%% ABSTRACT %%%%%%
%%%%%%%%%%%%%%%%%%%%%%
%
\begin{abstract}
	In this note, a Wegner estimate for random divergence-type operators that are monotone in the randomness is proven. 
	The proof is based on a recently shown unique continuation estimate for the gradient and the ensuing eigenvalue liftings. 
	
	The random model which is studied here contains quite general random perturbations, among others, 
	some that have a non-linear dependence on the random parameters.
\end{abstract}

\maketitle
%
%%%%%%%%%%%%%%%%%%%%%%%%%%%%%%
%%%%%%%% INTRODUCTION %%%%%%%%
%%%%%%%%%%%%%%%%%%%%%%%%%%%%%%
%
\section{Introduction}

Random divergence-type operators, i.e.~second order elliptic operators where the second order term is random, were studied in, e.g.,~\cite{FigotinK-96, Stollmann-98, DickeV-20}.
The interest in these operators stems from the study of propagation of classical and electromagnetic waves in random media.
Here, as for random Schr{\"o}dinger operators, one suspects that disorder leads to localization. 
In fact, Anderson localization for random divergence-type operators was proven in~\cite{FigotinK-96, Stollmann-98} for certain random models.
The proofs in both papers are based on the so-called multi-scale analysis, a tool that was developed to prove localization for 
random Schr{\"o}dinger operators, cf., e.g.,~\cite{FroehlichS-83,HoldenM-84}.
In order to start it, one needs to prove an initial length scale and a Wegner estimate.
Historically, in~\cite{FigotinK-96, Stollmann-98}, the authors relied on a strict covering condition that was imposed on the perturbation to prove 
those two estimates. 

In the context of random Schr{\"o}dinger operators, it turns out that initial length scale and Wegner estimates can be deduced from suitable eigenvalue liftings.
As was observed in~\cite{DickeV-20}, this is also true for Wegner estimates in the case of random divergence-type operators.
Indeed, in the last mentioned paper, eigenvalue liftings for these operators were obtained and applied to prove a Wegner estimate that does not rely
on a covering condition.

The proof of the eigenvalue liftings in~\cite{DickeV-20} heavily depends on a quantitative unique continuation 
estimate \emph{for the gradient} of eigenfunctions of divergence-type operators. 
Since it is well known that unique continuation may fail if the operator has merely H{\"o}lder continuous coefficients, cf.~\cite{Plis-63, Miller-73, Mandache-98},
one is thus restricted to certain Lipschitz continuous perturbations. 

%
%%%%%%%%%%%%%%%%%%%%%%%%%%%%%%
%%%%%%%% BEGIN FIGURE %%%%%%%%
%%%%%%%%%%%%%%%%%%%%%%%%%%%%%%
%
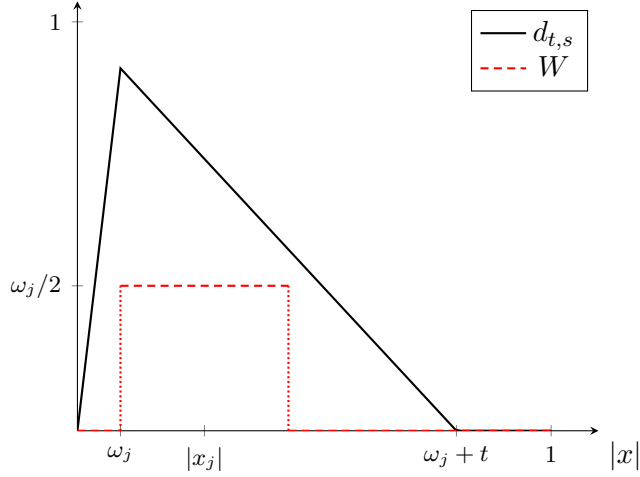
\begin{figure}
	\centering
	\begin{tikzpicture}
		\newcommand{\vartin}{10/8}
		\newcommand{\vart}{8/10}
		\newcommand{\varsin}{11}
		\newcommand{\vars}{1/11}
		\newcommand{\varrad}{.25*(\vart-\vars)}
		\newcommand{\varxzero}{\vars+\varrad}
		\begin{axis}[xlabel=$|x|$,xmin=0,xmax=1.1,ymin=0,ymax=1.05,axis x line=center,axis y line=center,ytick={{2*\varrad},1},xtick={0,\vars,\varxzero,\vart,1},
				xticklabels = {$0$, $\omega_j$, $|x_j|$, $\omega_j+t$, $1$},yticklabels = {$\omega_j/2$, $1$},xlabel style={below right},ylabel style={above right}]

			\addplot[thick,domain=0:1,samples=100,variable=\x] plot ({\x},{max(0,1-(\vartin*abs(\x)))-max(0,1-(\varsin*abs(\x)))});
			
			\addplot[red,thick,densely dashed,domain={\varxzero-\varrad}:{\varxzero+\varrad},variable=\x] plot ({\x},{.5*(\vart-\vars)});
			\addplot[red,thick,densely dashed,domain=0:{\varxzero-\varrad},variable=\x] plot ({\x},0);
			\addplot[red,thick,densely dashed,domain={\varxzero+\varrad}:1,variable=\x] plot ({\x},0);
			\addplot[red,thick,densely dotted,domain=0:{.5*(\vart-\vars)},variable=\y] plot({\varxzero-\varrad},{\y});
			\addplot[red,thick,densely dotted,domain=0:{.5*(\vart-\vars)},variable=\y] plot({\varxzero+\varrad},{\y});
			
			\addlegendentry{$d_{t,s}$}
			\addlegendentry{$W$}
				
		\end{axis} 
	\end{tikzpicture}
	\caption{The difference $d_{t,s}:=v(\cdot/(\omega_j+t))-v(\cdot/\omega_j)$ and the lower bound $W:=\frac{\omega_j}{2}\indic_{B(x_0,\omega_j/4)}$.}
	\label{fig:monotonicity}
\end{figure}
%
%%%%%%%%%%%%%%%%%%%%%%%%%%%%
%%%%%%%% END FIGURE %%%%%%%%
%%%%%%%%%%%%%%%%%%%%%%%%%%%%
%
So far, Wegner estimates for random divergence-type operators were only proved for random perturbations depending linearly on the random parameters. 
In this paper, we study quite general random perturbations, especially some that depend on the random parameters in a \emph{non-linear} way.
In some sense, the missing linearity is replaced by a monotonicity in the randomness, see Section~\ref{sec:notation} below for a precise definition. 

Let us consider a particular non-linear model in order to illustrate the main result: 
Let $(\omega_j)_{j\in\ZZ^d}$ be a sequence of independent, uniformly distributed random variables on the interval $[1/4,3/4]$,
define the function $v(x):=(1-|x|)_+$ on $\RR^d$ and the operators

\[
	H_\omega:=-\div\le\lr 1+\sum_{j\in\ZZ^d} v((x-j)/\omega_j)\rr\grad\re
\]
in $L^2(\RR^d)$, whose dependence on the random variables $(\omega_j)_j$ is obviously non-linear.
However, the dilation of $v$ satisfies
\bes
	v(\cdot/(\omega_j+t))-v(\cdot/\omega_j)\geq \frac{\omega_j}{2}\,\indic_{B(x_j,\omega_j/4)},\quad t<1/4,
\ees
for some point $x_j\in\RR^d$, see the visualization in Figure~\ref{fig:monotonicity}.
According to Definition~\ref{def:div-monotone}, the operator is therefore \emph{monotone in the randomness}.

Denoting the restriction to a box of integer side length $L$ centered at the origin with Dirichlet boundary conditions by $H_\omega^L$,
the Wegner estimate for this special case reads as follows.

\begin{theorem}\label{thm:Wegner-special-case}
	There are constants $\tilde{\eps}>0$, $\Xi>1$, $\cK>0$ depending only on the dimension $d$, such that for all $0<E_-<E_+<\infty$, 
	all $L\in \NN$, all $\eps\in (0,\tilde{\eps}]$ and all $E>0$ satisfying $[E-3\eps,E+3\eps]\subset [E_-,E_+]$ we have
	\bes
		\EE\big(\Tr\big[\chi_{[E-\eps,E+\eps]}(H^L_\omega)\big]\big) 
		\leq \cK_d\,E_+^{d/2} \,\eps^{\le\Xi\cdot(1+|\log E_-|+E_+^{2/3})\re^{-1}} |\L_L|^2.
	\ees
\end{theorem}

The proof is the same as the one of our main result below, but here we even kept track of the dependence on $E_-$ and $E_+$, cf.~also Remark~\ref{rem:constant-Wegner-lifting} below.

\subsection{Outline}

The rest of this paper is organized as follows:
In the next Section~\ref{sec:notation}, we introduce the notation and the general random model studied in this article. 
Thereafter, in Section~\ref{sec:main-result}, we formulate our main result, Theorem~\ref{thm:Wegner}, and discuss some properties of our random model. 
The subsequent Section~\ref{sec:proof-Wegner} is then dedicated to the proof of the main result. 
%
%%%%%%%%%%%%%%%%%%%%%%%%%%
%%%%%%%% NOTATION %%%%%%%%
%%%%%%%%%%%%%%%%%%%%%%%%%%
%
\section{Notation and the random model}\label{sec:notation}

For $L>0$ we denote by $\L_L=(-L/2,L/2)^d$ the cube of side length $L$ centered at the origin. 
Let $A=(a_{j,k})_{j,k=1,\dots,d}\colon\L_L\to\mathrm{Sym}(\RR^d)$ be a matrix function. 
Then $A$ is uniformly elliptic on $\L_L$ if there is a constant $\vartheta_E\geq 1$ such that
\be
	\vartheta_E^{-1}|\xi|^2\leq \xi\cdot A(x)\xi \leq \vartheta_E|\xi|^2,\quad x\in\L_L,\,\xi\in\RR^d.
	\label{eq:Elliptic} \tag{Ellip} 
\ee
We will often assume that $A$ is Lipschitz continuous, i.e., there is a constant $\vartheta_L\geq 0$ such that
\be
	\norm{A(x)-A(y)}_\infty\leq \vartheta_L|x-y|,\quad x,y\in\L_L,
	\label{eq:Lipschitz} \tag{Lip}
\ee
and that $A$ satisfies
\be
	\forall j\neq k,\,x\in\overline{\L_L}\cap\overline{(\L_L+L e_k)}\colon a_{j,k}(x)=a_{k,j}(x)=0.
	\tag{\text{Dir}} \label{eq:DIR-assumption}
\ee

Given a matrix function $A$ that satisfies~\eqref{eq:Elliptic}, we will denote 
the unique self-adjoint operator associated to the form 
\bes
	\hh^L\colon H^1_0(\L_L)\times H^1_0(\L_L) \to \CC,\quad
	(u,v)\mapsto \int_{\L_L} \grad u \cdot A\overline{\grad v}
\ees
by $H^L(A)\colon L^2(\L_L)\supseteq\cD(H^L(A))\to L^2(\L_L)$.
It is well known that $H^L(A)$ has compact resolvent and therefore purely discrete spectrum. 
We will denote its eigenvalues by $(E_n^L(A))_{n\in\NN}$, enumerated non-decreasingly and counting multiplicities, and its spectral projector 
associated with an interval $[a,b]\subset\RR$ by $\chi_{[a,b]}(H^L(A))$.

Next, we define the stochastic process and the model we are working with.
 
\begin{stochprocess}{(P)} \label{mod:randomness}
	Let $(\omega_j)_{j\in\ZZ^d}$ be a sequence of independent random variables with probability distributions $(\kappa_j)_{j\in\ZZ^d}$ 
	that have uniformly bounded densities $(g_j)_{j\in\ZZ^d}$ satisfying $\supp g_j\subset [\omega_-,\omega_+]\subset [0,1)$. 
	We will denote by $J>0$ the uniform upper bound of the densities, that is $\norm{g_j}_\infty\leq J$ for all $j\in\ZZ^d$. 
	Moreover, let $u_j\colon[0,1]\times\RR^d\to[0,\infty)$, $j\in\ZZ^d$, be a sequence of non-negative, 
	bounded and measurable functions for which we assume that
	\begin{enumerate}[(a)]
		\item each $u_j$ satisfies
		\bes
			0\leq u_j(t,\cdot)\leq M\,\indic_{G\L_1(j)}(\cdot)
		\ees
		for some constants $M,G>0$ and all $j\in\ZZ^d$, $t\in [0,1]$,
		\label{enum:assumption-bounded}
		\item there exist constants $\alpha,\beta>0,p,q\geq 0$ such that for all $0\leq s< t\leq 1$ there is some point $x_0\in\RR^d$ (depending on $s,t$ and $j$) 
		such that 
		\bes
			u_j(t,\cdot)-u_j(s,\cdot)\geq \alpha (t-s)^p\, \indic_{B(x_0,\beta (t-s)^q)} 
		\ees
		and 
		\label{enum:assumption-monoton}
		\item for all $t\in[\omega_-,\omega_+]$ the functions $u_j(t,\cdot)$, $j\in\ZZ^d$, are Lipschitz continuous 
		with Lipschitz constant at most $K>0$.
		\label{enum:assumption-Lip}
	\end{enumerate} 
	We consider the stochastic process 
	\be
		V_\omega(x) := \sum_{j\in\ZZ^d} u_j(\omega_j,x),\quad x\in\RR^d.
		\label{eq:perturbation}
	\ee
\end{stochprocess} 

With this process at hand we are in the position to introduce the family of operators we are concerned with in this paper. 

\begin{definition}[Divergence-type operators monotone in the randomness] \label{def:div-monotone}
	Fix a matrix function $A$ that satisfies~\eqref{eq:Elliptic},~\eqref{eq:Lipschitz} and~\eqref{eq:DIR-assumption} and consider the Stochastic Process~\ref{mod:randomness}.
	Then the family of random divergence-type operators
	\bes
		H_\omega^L := H^L(A+V_\omega\Id)
	\ees
	is said to be monotone in the randomness.
\end{definition} 

We also recall the notion of equidistributed sequences studied 
in, e.g., \cite{RojasMolinaV-13, TaeuferV-15, NakicTTV-18}.
\begin{definition} \label{def:equidistributed-seq}
	Given $G>0$ and $\delta\in (0,G/2)$, we say that a sequence $Z=(z_j)_{j\in (G\ZZ)^d}\subset \RR^d$ is $(G,\delta)$-equidistributed, 
	if $B(z_j,\delta)\subset \L_G(j)$ for all $j\in (G\ZZ)^d$.
	We set
	\bes
		S_{\delta,Z}(L):=\bigcup_{\substack{j\in(G\ZZ)^d \\ \L_G(j)\subset\L_L}} B(z_j,\delta).
	\ees
\end{definition}
%
%%%%%%%%%%%%%%%%%%%%%%%%%%%%%
%%%%%%%% MAIN RESULT %%%%%%%%
%%%%%%%%%%%%%%%%%%%%%%%%%%%%%
%
\section{Main result and discussion} \label{sec:main-result}

We henceforth assume that $H_\omega^L$ is a general random divergence-type operator monotone in the randomness as introduced in Definition~\ref{def:div-monotone}.
The next theorem is our main result. 

\begin{theorem}[Wegner estimate] \label{thm:Wegner}
	Let $0<E_-<E_+<\infty$. 
	There exist constants $C_W,\tilde{\eps}>0$, $\tau>1$, depending only on $\alpha$, $\beta$, $\vartheta_E$, $\vartheta_L$, $M$, $G$, $p$, $q$, $E_-$ and $E_+$,
	such that for all $L\in G\NN$, all $\eps\in (0,\tilde{\eps}]$ and all $E>0$ 
	satisfying $[E-3\eps,E+3\eps]\subset [E_-,E_+]$ we have
	\bes
		\EE\lr\Tr\le\chi_{[E-\eps,E+\eps]}(H_\omega^L)\re\rr \leq C_W \,\eps^{1/\tau} |\L_L|^2.
	\ees
\end{theorem}

In contrast to Wegner estimates for random Schr\"odinger operators that were proven in, e.g.,~\cite{HundertmarkKNSV-06,NakicTTV-18}, 
our result has a quadratic (and thus not optimal) dependence on the volume of the cube $\L_L$.
This is due to the fact that certain spectral shift estimates, that were proven in~\cite{HundertmarkKNSV-06} for Schr{\"o}dinger operators, are not 
available for divergence-type operators.

Additionally, it should be noted that in Theorem~\ref{thm:Wegner} we not only need to remove high energies, but
also energies close to zero. 
The reason for that is, that zero is not a spectral fluctuation boundary for the random divergence-type operators; cf.~the discussion after
Theorem~1.1 in~\cite{Stollmann-98} and the dependence on $E_-$ that was obtained in Theorem~\ref{thm:Wegner-special-case} above.

\subsection{Discussion}\label{ssec:discussion}

Let us discuss certain aspects of our random model: To begin with, 
it should be noted, that a more general stochastic process was introduced and studied in~\cite{NakicTTV-18} for the model of random Schr{\"o}dinger operators
$\tilde{H}_\omega^L:=(-\Delta+V_\omega)|_{\L_L}$, where $V_\omega$ is as in~\eqref{eq:perturbation} above.
More precisely, in that case condition~\eqref{enum:assumption-Lip} is not needed since 
there the stochastic process influences the model as a family of random potentials rather than as a family of \emph{random perturbations}.

This difference has a major impact on the theory of the associated random operators: 
Rather then relying on unique continuation estimates, as in the case of random Schr\"odinger operators, one relies 
on unique continuation estimates for the gradient. 
As was already said, unique continuation for elliptic second-order operators fails in general, if the coefficients are not regular enough. 
This is less important when working with random potentials since in that case the regularity of the $u_j$'s is of no interest. 
However, in our model, the regularity of the $u_j$'s directly influences the regularity of the coefficients of the second-order term and 
one therefore has to assume a certain regularity for them. 

Nevertheless, even in this somehow restricted case, 
random divergence-type operators with random perturbations depending non-linearly 
on the random variables have, to the best of the authors knowledge, not been studied before. 

\begin{remark}\label{rem:Lipschitz}
	The assumption of Lipschitz continuity is not uncommon in the literature on breather-type models.
	Indeed, this assumption was also needed in certain papers that dealt with random Schr{\"o}dinger operators with 
	breather-type potentials, cf., e.g.,~\cite{KirschV-10, SchumacherV-17}, in which the authors study Lifshitz tails, or~\cite{CombesHM-96,CombesHN-01}, 
	wherein Wegner estimates are given.
	In all the mentioned papers, these restrictions are necessary, because the given non-linear model is linearized in a 
	certain way to use methods originally developed for the case of linear models.
	
	In contrast, in the present paper, the Lipschitz continuity is needed for the unique continuation estimate we invoke and the techniques 
	used in the proof of the Wegner estimate are based on the monotonicity from assumption~\eqref{enum:assumption-monoton}.
\end{remark} 

For more detailed information on the history of, e.g., the standard random breather model, we refer to the discussion in the appendix of~\cite{NakicTTV-18}.

\begin{remark}
	The model studied in~\cite{DickeV-20} corresponds to a special case of Model~\ref{mod:randomness}, cf.~Example~\ref{ex:perturbations}\,(a) below.
	However, for the sake of accuracy, we note that slightly more general random variables were studied there.
	Note also that the considerations in~\cite{Stollmann-98} are not restricted to scalar valued multiples of the identity matrix and are therefore not in scope of our model. 
	For more details we refer to the respective article.
\end{remark} 

\subsection{Examples}

We conclude this section by showing that the class of random perturbations feasible by the methods in this paper is non-void. 
To this end, we formulate two examples:

\begin{example}\label{ex:perturbations} 
	Let $v_j\colon\RR^d\to[0,\infty)$, $j\in (G\ZZ)^d$, be uniformly bounded, Lip\-schitz 
	continuous functions satisfying $\Lip(v_j)\leq K'$ for some constant $K'\geq 0$. 
	
	(i) Assume that $v_j\geq\eta\indic_{B(z_j,\delta)}$ for some $\eta>0$ and some $(G,\delta)$-equi\-distributed sequence $Z=(z_j)_{j\in (G\ZZ)^d}$.
		Set $u_j(t,x)=tv_j(x)$. 
		Then conditions~\eqref{enum:assumption-bounded},~\eqref{enum:assumption-monoton} and~\eqref{enum:assumption-Lip} are satisfied for the $u_j$'s
		and this choice corresponds to the case of alloy-type random perturbations.
		
	(ii) Define $u_j(t,x):=v_j(x/t)$ and $u(0,x):=0$. 
		Moreover, assume that $\{u_j(t,\cdot)\colon t\in[0,1]\}$ satisfies~\eqref{enum:assumption-monoton}. 
		Then, in general, we must have $\omega_->0$ in order to guarantee the Lipschitz assumption~\eqref{enum:assumption-Lip}.
		This choice corresponds to the case of general random breather perturbations.

	Note that the required Lipschitz continuity does not permit us to study the (probably most interesting) case $v_j=\indic_{B(z_j,\delta)}$ 
	and the assumption $\omega_->0$ in (ii) also rules out some cases of the general random breather model, where one is particularly interested in the case $\omega_-=0$. 
\end{example} 

For a particular choice of $(v_j)$ the assumptions in the previous example are indeed satisfied.

\begin{example}\label{ex:explicit}
	We set $v_j(x):=v(x-j)$, $j\in (G\ZZ)^d$ and chose $v(x)=v_r(x)=(1-|x/r|)_+$ for some $r\in (0,1)$.
	Then $v$ is Lipschitz continuous with $\Lip(v)=1/r$ and we have $v(x)\geq \frac{1}{2}\indic_{B(0,r/2)}$; 
	hence, the requirements of Example~\ref{ex:perturbations}\,(i) are satisfied. 
	
	We can also consider the situation of Example~\ref{ex:perturbations}\,(ii) for this choice of $v$.
	To this end, note that $v(\cdot/t)$ has Lipschitz constant $1/(tr)$ and thus $u_j(\omega_j,\cdot)=v(\cdot/\omega_j)$ 
	has Lipschitz constant at most $1/(\omega_-r)$ if $\omega_->0$.
	This also shows why we necessarily need to assume that $\omega_->0$. 
\end{example} 

For more examples we refer the reader to the appendix of~\cite{NakicTTV-18}.
%
%%%%%%%%%%%%%%%%%%%%%%%%%%%%%%%%%
%%%%%%%% PROOF OF WEGNER %%%%%%%%
%%%%%%%%%%%%%%%%%%%%%%%%%%%%%%%%%
%
\section{Proof of the Wegner estimate} \label{sec:proof-Wegner}

This section is devoted to the proof of the main result. 
The idea essentially goes back to~\cite{HundertmarkKNSV-06} and we will adapt it to suit our random operators.
In fact, one of the main ingredients is a suitable adaptation of an eigenvalue lifting from~\cite{DickeV-20}. 

\subsection{Proof of Theorem~\ref{thm:Wegner}}

Let $\mu_+:=1-\omega_+$ and let $Q$ be the set of indices such that $(\supp u_j)\cap\L_L=\emptyset$ for all $j\notin Q$; according to~\eqref{enum:assumption-bounded} that is
$Q:=\{j\in\ZZ^d\colon \L_G(Gj)\cap \L_L\neq\emptyset\}$ .
We set $e=(1,\dots,1)\in \{0,1\}^{\# Q}$ and 
\bes
	V_\omega^Q := \sum_{j\in Q} u_j(\omega_j,\cdot).
\ees
Note that $V_\omega^Q$ is the effective perturbation on $\L_L$ in the sense that $V_\omega-V_\omega^Q$ does not influence the operator $H_\omega^L$. 
Assumption~\eqref{enum:assumption-monoton} implies
\bes
	V_{\omega+\mu\cdot e}^Q-V_{\omega}^Q \geq \alpha\mu^p\,\sum_{j\in Q} \indic_{B(x_0(j), \beta\mu^q)},\quad \mu\leq \mu_+,
\ees
and since $B(x_0(j),\beta \mu^q)\subset G\L_1(j) =\L_G(Gj)$, there exists a $(G,\beta\mu^q)$-equi\-distributed sequence $Z=(z_k)_{k\in(G\ZZ)^d}$ such 
that
\bes
	\sum_{j\in Q} \indic_{B(x_0(j), \beta\mu^q)} \geq  \indic_{S_{\beta\mu^q,Z}(L)}.
\ees
Hence,
\be
	M \geq V_{\omega+\mu\cdot e}^Q-V_{\omega}^Q \geq \alpha\mu^p\, \indic_{S_{\beta\mu^q,Z}(L)},
	\label{eq:random-perturbation-difference}
\ee
since the random perturbation is uniformly bounded by $M$ and non-negative. 

We abbreviate $\lam_n^L(\omega):=E_n^L(A+V_\omega\Id)$, $n\in\NN$, and note that $A+V_\omega\Id$ is uniformly elliptic with ellipticity constant 
$\theta_E=\vartheta_E+M$, Lipschitz continuous with Lipschitz constant $\theta_L=\vartheta_L+K$ and satisfies~\eqref{eq:DIR-assumption}, 
since $V_\omega^Q$ is uniformly bounded and non-negative, Lipschitz continuous and the perturbation is a multiple of the identity matrix. 
In order to estimate the trace we choose a non-decreasing function $\rho=\rho_\eps\in\cC^\infty(\RR;[-1,0])$, $\eps>0$, satisfying $\norm{\rho'}_\infty\leq 1/\eps$,
$\rho\equiv -1$ on $(-\infty,-\eps)$ and $\rho\equiv 1$ on $(\eps,\infty)$ such that
\be
	\indic_{[E-\eps,E+\eps]}(x)\leq \rho(x +2\eps -E)-\rho(x-2\eps-E)
	\leq \indic_{[E-3\eps,E+3\eps]}(x) 
	\label{eq:smear-fct}
\ee
for $x\in\RR$. 
With this function at hand, the spectral theorem implies
\eq{
	\EE&\lr\Tr\le\chi_{[E-\eps,E+\eps]}(H_\omega^L)\re\rr\nonumber \\
	&\leq \EE\lr\sum_{n\in\NN}\le\rho(\lam_n^L(\omega)+4\eps -E -2\eps)-\rho(\lam_n^L(\omega)-E-2\eps)\re\rr
	\label{eq:smear}
}
and we further estimate~\eqref{eq:smear} using a refined version of the eigenvalue lifting from~\cite{DickeV-20}.
We rely on a refinement, since in the proof of Lemma~\ref{lem:Wegner-lifting} below
we need to handle perturbations that are not Lipschitz continuous and, at the same time, track the dependence on $\mu$. 

The result we use reads as follows.

\begin{lemma}\label{lem:eigenvalue-lifting}
	Let $B$ be a matrix function that satisfies~\eqref{eq:Elliptic},~\eqref{eq:Lipschitz} (with constants $\theta_E$ and $\theta_L$),~\eqref{eq:DIR-assumption} and
	let $G,\eta\geq 0$. 
	Then there are $\delta_0\in (0,G/2)$ and $N'>0$ depending only on $\theta_E$, $\theta_L$, $G$, $\eta$ and the dimension $d$, such that for all $L\in G\NN$, all
	$\delta\in(0,\delta_0]$, all $(G,\delta)$-equidistributed sequences $Z$, all $W\in L^\infty(\L_L)$ satisfying $W\geq \eta\delta\indic_{S_{\delta,Z}(L)}$,  all $0<E_-<E_+<\infty$
	and all $n\in\NN$ such that $E_-\leq E_n^L(B)\leq E_n^L(B+W\Id)\leq E_+$ we have
	\bes
		E_n^L(B+t\,W\Id)\geq E_n^L(B)+t\,E_-^2\eta\lr\frac{\delta}{4G}\rr^{N'\cdot (1+E_+^{2/3})},\quad t\in[0,1].
	\ees
\end{lemma} 

\begin{proof}
	There is a Lipschitz continuous function $\tilde{W}$ satisfying 
	\bes
		W\geq \eta\delta\indic_{S_{\delta,Z}(L)}\geq \tilde{W}\geq \eta\delta\indic_{S_{\delta/2,Z}(L)}
	\ees
	while having a Lipschitz constant $\Lip(\tilde{W})\sim \eta$. 
	Hence, applying~\cite[Corollary~6.5]{DickeV-20} with the matrix function $B$ and the perturbation $\tilde{W}$, there are $\delta_0\in (0,G/2)$ and $N>0$,
	depending only on $\theta_L$, $\theta_E$, $\eta$, $G$ and the dimension $d$, such that 
	\bes
		E_n^L(B+t\,\tilde{W}\Id)\geq E_n^L(B)+t\,E_-^2\eta\delta\lr\frac{\delta}{4G}\rr^{N\cdot (1+G^{4/3}E_+^{2/3})}
	\ees
	holds for all $t\in [0,1]$ (note that the second summand is smaller than the one we obtain from the corollary in~\cite{DickeV-20}).
	Using the Min-Max principle, we see that $E_n^L(B+t\,W\Id)\geq E_n^L(B+t\,\tilde{W}\Id)$, and this completes the proof when choosing $N'$ appropriately.
\end{proof} 

The lemma is used to prove the following statement.

\begin{lemma} \label{lem:Wegner-lifting}
	For all $n$ contributing to the right hand side of~\eqref{eq:smear} we have
	\bes
		\lam_n^L(\omega+\mu\cdot e) \geq \lam_n^L(\omega)+\mu^\tau,\quad\mu\leq\mu',
	\ees
	where $\mu'\in (0,1)$ is small enough (depending only on $\alpha$, $\beta$, $\vartheta_E$, $\vartheta_L$, $M$, $K$, $G$, $q$ and $p$) and
	$\tau>1$ is a large constant (depending only on $\alpha$, $\beta$, $\vartheta_E$, $\vartheta_L$, $M$, $K$, $G$, $p$, $q$, $E_-$ and $E_+$).
\end{lemma} 

\begin{proof}
	Let $B:=A+V_\omega^Q$ and $W':=\alpha\mu^p\indic_{S_{\beta\mu^q,Z}(L)}.$ 
	Inequality~\eqref{eq:random-perturbation-difference} implies 
	\bes
		V_{\omega+\mu\cdot e}^Q \geq W' + V_{\omega}^Q\quad\text{and}\quad  \lam_n^L(\omega+\mu\cdot e)\geq E_n^L(B+W'\Id)
	\ees
	by the Min-Max principle. 
	Hence, it suffices to show that there exists an appropriate constant $\tau>1$ depending only on 
	$\alpha$, $\beta$, $\vartheta_E$, $\vartheta_L$, $M$, $K$, $G$, $p$, $q$, $E_-$ and $E_+$, such that  
	\be
		E_n^{L}(B+W'\Id)\geq E_n^L(B)+\mu^\tau
		\label{eq:Wegner-lifting-proof}
	\ee
	for sufficiently small $\mu$.
	
	We aim to apply Lemma~\ref{lem:eigenvalue-lifting} to prove the latter estimate for all $n$ contributing to~\eqref{eq:smear}.
	To this end, we need to verify the assumptions of the last mentioned lemma.
	At first, we set $\delta=\beta\mu^q$, $\eta=\alpha/\beta$ and choose a slightly different function 
	\bes
		W:=\eta\delta\indic_{S_{\delta,Z}(L)} = \alpha\mu^q\indic_{S_{\beta\mu^q,Z}(L)}
	\ees
	that will act as our perturbation.
	Note that $W=\mu^{p-q}W'$.
	
	We verify that there are suitable upper and lower bounds for the energy interval in which the eigenvalues $E_n^L(B+t\,W\Id)$, $t\in[0,1]$, lie. 
	Observe that~\eqref{eq:smear-fct} and the assumption $[E-3\eps,E+3\eps]\subset [E_-,E_+]$ imply
	that only the eigenvalues $\lambda_n^L(\omega)\in[E_-,E_+]$ give a non-zero contribution in~\eqref{eq:smear}. 
	Thus, $E_n^L(B)=\lambda_n^L(\omega)\geq E_-$, which is the desired lower bound.
	For the upper bound we estimate
	\be
		E_n^L(B+W\Id)\leq E_n^L(B)+\alpha E_n^L(\Id) \leq E_++\alpha E_n^L(\Id).
		\label{eq:upper-bound-Wegner-lifting-est}
	\ee
	To further estimate the right hand side, we notice that $(E_n^L(\Id))_n$ are the eigenvalues of the Dirichlet-Laplacian on the cube $\L_L$. 
	For those, according to~\cite[p.~266]{ReedS-78}, it holds 
	\bes
		\#\{n\colon E_n^L(\Id)\leq \tilde{E}\} = \#\{m\in\NN^d\colon |m|<(2/\pi)L\tilde{E}^{1/2}\}
	\ees
	for all $\tilde{E}\geq 0$.
	This identity shows that for some constants $\cK_1,\cK_2>0$ depending only on 
	the dimension $d$, the lower bound $n(\tilde{E})$ and the upper bound 
	$N(\tilde{E})$ for the number of eigenvalues below $\tilde{E}$ satisfy 
	\be
		\cK_1 \tilde{E}^{d/2}L^d\leq n(\tilde{E})\leq N(\tilde{E})\leq\cK_2 \tilde{E}^{d/2}L^d.
		\label{eq:est-number-elements-Laplacian}
	\ee
	
	However, as already noted above, only the eigenvalues $\lambda_n^L(\omega)\in[E_-,E_+]$ give a non-zero contribution and we clearly have $\lambda_n^L(\omega)>E_+$ 
	if $\lambda_n^L(\omega)\geq \vartheta_E^{-1}E_n^L(\Id)>E_+$.
	At his point, the upper bound $N(\tilde{E})$ with $\tilde{E}=\vartheta_EE_+$ shows, that 
	we have $E_n^L(\Id)\geq \vartheta_EE_+$ if $n\geq N_0:=N(\vartheta_EE_+)$.
	Thus, at most the first $N_0$-summands in~\eqref{eq:smear} are non-zero and 
	\eqs{
		\EE&\lr\sum_{n\in\NN}\le\rho(\lam_n^L(\omega)+4\eps -E -2\eps)-\rho(\lam_n^L(\omega)-E-2\eps)\re\rr\nonumber \\
		&=\EE\lr\sum_{n=1}^{N_0}\le\rho(\lam_n^L(\omega)+4\eps -E -2\eps)-\rho(\lam_n^L(\omega)-E-2\eps)\re\rr.
	}
	On the other hand, using~\eqref{eq:est-number-elements-Laplacian}, there is some $E'_{+}$ such that 
	$n(E'_{+})>N_0$ and $E'_{+}$ can be chosen such that it satisfies $E'_{+}\leq \cK_3\cdot(1+E_+)$ for some constant $\cK_3$ 
	depending only on $\vartheta_E$ and on the dimension $d$.
	Moreover, since $n(E_+')$ is the lower bound for the number of eigenvalues $E_n^L(\Id)$ below $E_+'$, this yields $E_n(\Id)\leq E'_+$ for all $n\leq N_0$.
	Combining the latter with~\eqref{eq:upper-bound-Wegner-lifting-est} shows 
	\bes
		E_n^L(B+W\Id)\leq E_++\alpha E'_+\leq \cK_4(1+E_+)=:E_{++}
	\ees
	for some constant $\cK_4$ depending only on $\vartheta_E$, $\alpha$ and on the dimension $d$.
	
	We have successfully verified the assumption of Lemma~\ref{lem:eigenvalue-lifting} with $E_+$ replaced by $E_{++}$.
	Consequently, there are constants
	$\delta_0\in (0,G/2)$ and $N'>0$ depending only on $\alpha$, $\beta$, $\theta_E$, $\theta_L$, $G$, $M$, $p$, $q$ and the dimension $d$, such that
	\be
		E_n^L(B+t\,W\Id)\geq E_n^L(B)+t\,E_-^2\eta\lr\frac{\delta}{4G}\rr^{N'\cdot (1+E_{++}^{2/3})},\quad t\in[0,1],
		\label{eq:eigenvalue-lifting-Wegner}
	\ee
	provided that $\mu'\in (0,1)$ is chosen such that $\delta\leq\delta_0$ for $\mu\leq\mu'$.
	In order to conclude the proof we recall that $W'=\mu^{p-q}W$ and distinguish between two cases: 
	\begin{enumerate}[(i)]
		\item If $p\leq q$, then $\mu^{p-q}\geq 1$. Hence $W'\geq W$ and~\eqref{eq:eigenvalue-lifting-Wegner} with $t=1$ implies
		\bes
			E_n^L(B+W'\Id)\geq E_n^L(B)+E_-^2\eta\lr\frac{\beta\mu^q}{4G}\rr^{N'\cdot (1+E_{++}^{2/3})}.
		\ees
		\item If $p>q$, then $t:=\mu^{p-q}\in (0,1)$ and $W'=tW$. 
		With this choice of $t$ inequality~\eqref{eq:eigenvalue-lifting-Wegner} implies
		\bes
			E_n^L(B+W'\Id)\geq E_n^L(B)+\mu^{p-q}\,E_-^2\eta\lr\frac{\beta\mu^q}{4G}\rr^{N'\cdot (1+E_{++}^{2/3})}.
		\ees
	\end{enumerate} 
	Since $E_{++}=\cK_4(1+E_+)$, we can go back to $E_+$ by appropriately adapting the constant $N'$ and bringing everything together, 
	we have thus shown that~\eqref{eq:Wegner-lifting-proof} holds for some appropriate constant $\tau>1$ depending only on 
	$\alpha$, $\beta$, $\vartheta_E$, $\vartheta_L$, $M$, $K$, $G$, $p$, $q$, $E_-$ and $E_+$.
	This finishes the proof of the lemma.
\end{proof} 

\begin{remark} \label{rem:constant-Wegner-lifting}
	For some $\tilde{N}>0$ depending only on $\alpha$, $\beta$, $\vartheta_E$, $\vartheta_L$, $M$, $K$, $G$, $p$, and $q$, 
	the constant $\tau$ can be chosen as $\tau=\tilde{N}\cdot(1+E_+^{2/3}+|\log E_-|)$.
	This gives rise to the explicit constant in Theorem~\ref{thm:Wegner-special-case} above. 
\end{remark} 

With the lemma at hand, a simple calculation shows that with $\tilde{\eps}=(\mu')^\tau/4$ we have
\be
	\lam_n^L(\omega+\eps'\cdot e) \geq \lam_n^L(\omega)+4\eps\quad\text{if}\quad\eps'=( 4\eps)^{1/\tau}
	\label{eq:eigenvalue-lifting-choice-of-t}
\ee 
for all $\eps\leq\tilde{\eps}$.
Note that the definition of $\tilde{\eps}$ and the choice of $\eps'$ implies $\eps'\leq\mu'$.
We use inequality~\eqref{eq:eigenvalue-lifting-choice-of-t} to further estimate
\eqs{
	\EE\,\Bigg(\sum_{n\in\NN}&\le\rho(\lam_n^L(\omega)+4\eps -E -2\eps)-\rho(\lam_n^L(\omega)-E-2\eps)\re\Bigg)\\
	&\leq \EE\,\Bigg(\sum_{n\in\NN}\le\rho(\lam_n^L(\omega+\eps'\cdot e)-E-2\eps)-\rho(\lam_n^L(\omega)-E-2\eps)\re\Bigg)\\
	&=\EE\lr\Tr\le\rho(H^L_{\omega+\eps'\cdot e}-E-2\eps)-\rho(H^L_\omega-E-2\eps) \re \rr.
}

From this point forward the proof is essentially the same as the one given in, e.g.,~\cite{HundertmarkKNSV-06,NakicTTV-18}. 
For the sake of completeness, we nevertheless give the details: 
Let $r=(r_j)_{j=1,\dots,\#Q}$ be an enumeration of the lattice points $Q$.
For $\ell\in\NN$, $\mu\leq\mu_+$ and $s\in[\omega_-,\omega_+]$ we define 
\bes
	\nu^{(\ell)}(\mu,s)=\lr\nu^{(\ell)}_j(\mu,s)\rr_{j=1,\dots,\#Q}
\ees
recursively by 
\eqs{
	\nu^{(1)}_j(\mu,s):=\begin{cases}
		s & \text{if} \quad j = r_1\\
		\omega_j & \text{if} \quad j \neq r_1
	\end{cases},
	\enspace
	\nu^{(\ell)}_j(\mu,s):=\begin{cases}
		s & \text{if}\quad  j = r_\ell\\
		\nu^{(\ell-1)}_j(\mu,\omega_j+\mu) & \text{if}\quad  j \neq r_\ell
	\end{cases}
}
for each $j\in\{1,\dots\#Q\}$ and $\ell\in\{2,\dots,\#Q\}$. 
Then
\eqs{
	\EE(&\Tr\le\rho(H^L_{\omega+\eps'\cdot e}-E-2\eps)-\rho(H^L_\omega-E-2\eps) \re)\nonumber\\
	&\leq \sum_{\ell=1}^{\#Q} \EE\lr \Tr\lr\rho(H^L_{\nu^{(\ell)}(\eps',\omega_{r_\ell}+\mu)}-E-2\eps)
		-\rho(H^L_{\nu^{(\ell)}(\eps',\omega_{r_\ell})}-E-2\eps)\re\rr\nonumber \\
	&=\sum_{\ell=1}^{\#Q} \EE\lr\Phi_\ell(\omega_{r_\ell}+\eps')-\Phi_\ell(\omega_{r_\ell})\rr
	=: I
}
where 
\bes
	\Phi_\ell(s):=\Tr\le\rho(H^L_{\nu^{(\ell)}(\eps',s)}-E-2\eps)\re<0.
\ees

We estimate each summand of $I$ separately. 
To this end, we note that 
\eqs{
	\EE\lr\Phi_\ell(\omega_{r_\ell}+\eps')-\Phi_\ell(\omega_{r_\ell})\rr& \\
	=\EE^{Q\setminus\{r_\ell\}}&\lr \int \Phi_\ell(\omega_{r_\ell}+\eps')-\Phi_\ell(\omega_{r(\ell)})\Diff{\kappa_{r_\ell}(\omega_{r_\ell})} \rr,
}
where $\EE^{Q\setminus\{r_\ell\}}$ denotes the expectation with respect to all $\omega_j$, $j\in Q\setminus\{r_\ell\}$.
In order to estimate the right hand side, we invoke the following lemma, which is an easy consequence of, e.g.,~\cite[Lemma~6]{HundertmarkKNSV-06}. 

\begin{lemma}
	Let $0<\omega_-<\omega_+<1$, $\Phi\colon\RR\to (-\infty,0]$ be a non-decreasing, bounded function and assume that $\kappa$ 
	is a probability distribution with bounded density $g$ satisfying $\supp g\subset [\omega_-,\omega_+]$. 
	Then
	\bes
		\int_\RR \Phi(\lambda+\gamma)-\Phi(\lambda)\Diff{\kappa(\lambda)} \leq -\gamma \norm{g}_\infty \Phi(\omega_-)\quad\text{for all}\quad \gamma>0.
	\ees
\end{lemma} 

Since the operators $H^L_{\nu^{(\ell)}(\eps',s)}$, $s\in [\omega_-,\omega_+]$, $\ell\in\{1,\dots,\#Q\}$, are non-negative and the real-valued
function $\Phi_\ell(\cdot)$ is bounded and non-decreasing, the assumptions of the lemma are satisfied for $\Phi=\Phi_\ell$ and $\kappa=\kappa_j$. 
Hence
\bes
	\int \Phi_\ell(\omega_{r_\ell}+\eps')-\Phi_\ell(\omega_{r(\ell)})\Diff{\kappa_{r_\ell}(\omega_{r_\ell})} \leq -J \eps' \Phi_\ell(\omega_-).
\ees
We estimate
\bes
	-\Phi_\ell(\omega_-) = \sum_{n\in\NN}(-\rho_\eps)(E_n^L(H^L_{\nu^{(\ell)}(\eps',\omega_-)})-E-2\eps) 
	\leq \#\{ E_n^L(H^L_{\nu^{(\ell)}(\eps',\omega_-)}) \leq E_+\}
\ees
and by Weyl asymptotics there is a constant $C_{E_+}>0$, depending only
on the dimension $d$, the ellipticity constant $\vartheta_E$ and $E_+$, such that 
\bes
	\#\{ E_n^L(H^L_{\nu^{(\ell)}(\eps',\omega_-)}) \leq E_+\}\leq C_{E_+} L^d.
\ees
We have thus shown that 
\be
	\EE\lr\Phi_\ell(\omega_{r_\ell}+\eps')-\Phi_\ell(\omega_{r_\ell})\rr \leq C_{E_+}J\eps' L^d.
	\label{eq:upper-bound-expectation}
\ee

In order to conclude the proof we first estimate each summand of $I$ using~\eqref{eq:upper-bound-expectation} which provides us with
\bes
	\EE\lr\Tr\le\chi_{[E-\eps,E+\eps]}(H_\omega^L)\re)\rr \leq C_{E_+}J \eps'  L^d (\#Q).
\ees
Now, we use that $\#Q\leq \cK_5L^d$ for some constant $\cK_5>0$ depending only on the dimension $d$ and that 
$\eps'=( 4\eps)^{1/\tau} \leq \cK_6\eps^{1/\tau}$
for some constant $\cK_6>0$, depending only on $\alpha$, $\beta$, $\vartheta_E$, $\vartheta_L$, $M$, $K$, $J$, $G$, $p$, $q$, $E_+$ and $E_-$.
Bringing everything together, we have thus proven the theorem.\hfill\qedsymbol
%
%%%%%%%%%%%%%%%%%%%%%%%%%%%%%%%%%
%%%%%%%% ACKNOWLEDGMENTS %%%%%%%%
%%%%%%%%%%%%%%%%%%%%%%%%%%%%%%%%%
%
\section*{Acknowledgments}

The author thanks I.~Veseli\'c for suggesting this field of research as well as A.~Seelmann and I.~Veseli\'c for helpful comments on an earlier version of this manuscript.

%
%%%%%%%%%%%%%%%%%%%%%%%%%%%%%%%%%
%%%%%%%%% BIBLIOGRAPHY %%%%%%%%%%
%%%%%%%%%%%%%%%%%%%%%%%%%%%%%%%%%
%
\newcommand{\etalchar}[1]{$^{#1}$}
\def\polhk#1{\setbox0=\hbox{#1}{\ooalign{\hidewidth
  \lower1.5ex\hbox{`}\hidewidth\crcr\unhbox0}}}

\end{document}